\documentclass[10pt,twocolumn,twoside]{IEEEtran}

\usepackage{graphicx,cite}
\usepackage[cmex10]{amsmath}

\usepackage{amssymb,bm}
\usepackage{array}
\usepackage{amsfonts}

\usepackage{color}

\usepackage{ifthen}

\usepackage{setspace}
\usepackage{amsmath,amsfonts}
\usepackage{amssymb}
\usepackage{graphicx}

\newcommand{\RemoveThis}{\ifthenelse{1< 0}}

\newcommand{\qed}{\hfill $\square$}

\newtheorem{lemma}{\textbf{Lemma}}

\newtheorem{theorem}{\textbf{Theorem}}
\newtheorem{remark}{\textbf{Remark}}

\begin{document}

\title{On Threshold Routing in a Service System with Highest-Bidder-First and FIFO Services}
\author{
  \begin{tabular}{ccc}
    Tejas Bodas  & D.~Manjunath \\
    IIT Bombay, INDIA & IIT Bombay, INDIA 
  \end{tabular}
}

\maketitle

\begin{abstract}
In this paper, we consider a two server system serving heterogeneous customers. One of the server has a FIFO scheduling policy and charges a fixed admission price to each customer. The second queue follows the highest-bidder-first (HBF) policy where an arriving customer bids for its position in the queue. Customers make an individually optimal choice of the server and for such system, we characterize the equilibrium routing of customers. We specifically show that this routing is characterized by two thresholds.
\end{abstract}

\section{Introduction}
\label{sec:intro}

Consider two make-to-order firms manufacturing an identical product.
Upon receiving an order, the firms must assemble the product and 
deliver it to the customer ordering it. Each firm can 
assemble only one quantity at a time and the time taken to assemble the 
product need not be deterministic.
During an assembly of a product, if there are more orders being placed 
by other customers, then these orders have to be fulfilled by the 
firm by suitably scheduling the subsequent orders. The two firms differ
in their pricing strategy and the scheduling policy for choosing subsequent orders.
One of the firm charges a fixed admission price for the product and 
maintains a FIFO scheduling discipline. The second firm employs a bidding
policy where subsequent customers place a bid and their queue position
in the schedule in proportional to the bid placed.
The customers ordering the product differ in their cost for unit delay
and are thus sensitive to the delay in receiving the product. 
When placing an order, the customer does 
not know the number of pending orders but may be informed about 
the service rate and the arrival rate for the orders. When 
ordering the product, the customers have to decide which firm to choose
and if they choose the bidding firm, then what is the optimal bid to 
be made such that the cost of obtaining the product (the sum of the 
monetary and the delay cost) is minimized. Motivated by this problem, 
our interest is to characterize the equilibrium choice of the firm
made by the heterogeneous customers. 

Applicable to more general setting, a formal description of the problem
considered in this paper is as follows. Consider a two server service 
system with customers arriving according to a homogeneous Poisson 
process. The customers are heterogeneous in their cost for 
unit delay. The service system
consists of two servers and the customers are required to obtain service at 
one of these two servers. There is no dispatcher available to route
the customers and hence each arriving customer has to make an individually
optimal queue join decision. Each server in the service system has an 
associated queue and the two queues differ from each other in their 
scheduling policy. One of the queue has the standard FIFO scheduling 
policy and to monetize the offered service, it charges a fixed admission price
to its customers. The other queue has a non-preemptive priority scheduling
discipline where after the current service completion, a customer with the
highest priority level is next chosen for service from the pool of customers
waiting in the queue. The priority of a customer in this queue is determined by
the bid paid by each arriving customer. Naturally, a higher bid corresponds to a higher 
priority in the queue. Such a scheduling policy is also known as the 
highest-bidder-first (HBF) policy and was introduced by 
Kleinrock \cite{Kleinrock67}. In this paper, our primary interest is to 
characterize the equilibrium routing satisfying the Wardrop conditions 
\cite{Wardrop52} and determine the bidding decision made by those customers
choosing the HBF server.

Such a system with parallel HBF and FIFO services was first analyzed in
\cite{Bodas14b}. To investigate the effect on the revenue from an HBF 
server, a free FIFO service was introduced in the system. Further
a minimum bid was made mandatory for those choosing the HBF server.
For such a system, the equilibrium routing and bidding 
strategy was first analyzed in \cite{Bodas14b}. It was shown that the Wardrop equilibrium 
routing is characterized by a single threshold and customers with 
delay sensitivity (cost per unit delay) above the threshold 
choose the HBF server while the rest choose the FIFO server. Two scenarios
were considered for modeling the system; in the first scenario, a free
FIFO server was added in parallel to an existing HBF server. In the second 
scenario, the total service capacity was shared between the HBF and the 
FIFO server. Assuming that the customers cannot balk, it was shown that 
an addition of a free FIFO server decreases the system revenue.
On the contrary, with the help of numerical examples, it was conjectured 
that sharing capacity with a FIFO server improves the revenue from
the HBF server. 

The primary difference between the system model considered in this paper 
and that of \cite{Bodas14b} is as follows.
We assume that the FIFO server is not free but in fact comes with
an admission price. This assumption makes the model more naturally 
applicable to a variety of revenue based service systems such as  
the above example for make-to-order firms. We relax the
assumption of a minimum bid and analyze the equilibrium routing 
and bidding rule for this problem. As in \cite{Bodas14b}, the equilibrium bids 
in the HBF server (by those choosing HBF under equilibrium routing)
can be determined from the analysis in \cite{Glazer86,Lui85}. 
We begin by analyzing whether a single threshold 
routing function as in \cite{Bodas14b} satisfies the equilibrium 
routing conditions. To our surprise, this is not the case. 
We then check for the threshold routing where customers with 
sensitivity above a threshold choose the FIFO server while the 
rest choose the HBF server. We show that such a candidate for equilibrium 
routing also does not satisfy the necessary conditions for 
Wardrop equilibrium. In our main result, we prove that the
(Wardrop) equilibrium routing is of the following type:
there exists two threshold and customers with sensitivity 
between the two thresholds choose the FIFO server
while the rest choose the HBF server. To the best of 
our knowledge, the result is novel and has a useful insight.
While the `middle class' of the population (based on their sensitivities) 
choose the FIFO service, the remaining customers (specifically
those with high and low delay sensitivity) choose the HBF server.

The rest of the paper is organized as follows. In the 
next section, we shall introduce the notation and 
recall some preliminary results from \cite{Bodas14b}.
We prove our main results in Section \ref{sec:main}
and this is followed by a discussion summarizing the
results.

\section{Notation and Preliminaries}
\label{sec:prelims}
In this section, we shall introduce the basic notation 
and recall some preliminary results from \cite{Bodas14b}.
We assume as in \cite{Bodas14b} that the customers arrive to
the service system according to a homogeneous
Poisson process of rate $\lambda.$ Service times of customers are
i.i.d. random variables with distribution $G(\cdot)$ and unit mean.
Each arriving customer has an associated parameter $\beta$ 
which is a realization of the random variable $\bm{\beta},$ $0
\leq a \leq \bm{\beta} \leq b < \infty.$ $\bm{\beta}$ represents a customers
cost per unit delay. Let $F(\beta)$ denote the distribution of  
$\bm{\beta}$ which is assumed to be absolutely continuous in $(a,b).$ 
We also call $\bm{\beta}$ as the type of the customer and call
$F(\beta)$ as the type profile.

The first server uses the non preemptive HBF discipline and serves 
at rate $\mu_1.$
The second server uses the FIFO discipline and serves at rate $\mu_2.$
Customers choosing the HBF server will have to place a bid before 
joining its queue while those choosing the FIFO server have to pay a 
fixed admission price denoted by $c$. We will assume that all arrivals
will have to
receive service from one of the two servers and they cannot balk. Thus
an arriving customer now has to make the following decisions on
arrival; which server to use, and, if it chooses the HBF server, then
the value of its bid. As in \cite{Bodas14b}, we assume oblivious 
decisions and let $p(\beta): [a,b] \to [0,1]$ denote the probability
that a customer of type $\beta$ chooses the FIFO server. Further, let
$X(\beta)$ be the equilibrium bid if such a customer chooses the HBF
server. For a preliminary analysis of the HBF queue, refer 
\cite{Kleinrock67}. Lui \cite{Lui85} and Glazer and Hassin
\cite{Glazer86} were the first to consider 
the case with heterogeneous customers (characterized by $\bm{\beta}$)
and have determined the equilibrium bidding function $X(\beta)$. 
The function $X(\beta)$ determines the optimal value of that bid to
be made by a customer of type $\beta$ such that the sum of the bid 
and the expected waiting cost in the queue is minimized. Specifically, 
 it was shown that $X(\beta)$ is given by 
\begin{equation}
  \label{eq:bid_fx}
  X(\beta) = \int_0^\beta \frac{2 \rho W_0 y}{\left( 1-\rho+\rho
      F(y)\right)^3}\ dF(y)
\end{equation}
 where $\rho$ denotes the traffic intensity, $F(\cdot)$ denotes the 
 underlying distribution of $\bm{\beta}$ and $W_0$ denotes the expected
 waiting time in the HBF server added to that of an arriving customer
 due to the residual service time of an existing customer. This is given
 by, 
\begin{displaymath}
  W_0 = \frac{\lambda}{2} \int_0^\infty \tau^2 dG(\mu \tau)
\end{displaymath}
where $\lambda$ and $\mu$ denote the arrival rate of customers 
and the service rate of the HBF server respectively. It was further shown 
that for a customer of type $\beta,$ its expected 
waiting time $W(\beta)$ is given by  

\begin{equation}
\label{eq:Wbeta}
W(\beta) = \frac{\mu^2 W_0}{(\mu - \lambda(1-F(\beta)))^2}. 
\end{equation}
%
%
%

Now for a given $p(\beta),$ it is easy to 
see that the arrival rate to the FIFO server is $\lambda_2 := \lambda
\int_{0}^{\infty} p(\beta) dF(\beta)$ while the arrival rate to the HBF
server is $\lambda_1:= \lambda - \lambda_2$. Let $\rho_i := \lambda_i/\mu_i.$ 
A customer of type $\beta$ that chooses the HBF server experiences a bid-dependent 
waiting time that will be denoted by $W_1(\beta).$ 
%
The customers choosing
the FIFO server experience an expected waiting time denoted by
$W_2(\lambda_2).$ Continuing with the notation of \cite{Bodas14b}, let 
$D_1(\beta):= W_1(\beta) + \frac{1}{\mu_1}$ and $D_2 := W_2(\lambda_2) +
\frac{1}{\mu_2}$ be the expected sojourn times in, respectively, the
HBF and the FIFO servers. Refer Fig.~\ref{fig:2queue} for an 
illustration of the system model.
\begin{figure}
  \begin{center}
    \includegraphics[height=2.8in]{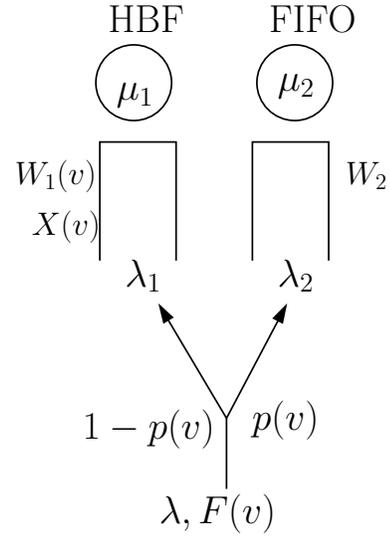}
  \end{center}
  \caption{System with a HBF server and a FIFO server.}
  \label{fig:2queue}
\end{figure}

In this paper, the primary interest is to obtain the equilibrium
strategy henceforth denoted by $(p^E(\beta), X^E(\beta))$. Note that
$X^E(\beta)$ needs to be determined for only those customers that
under equilibrium decide to join the HBF queue. Clearly, the 
system considered is non-atomic and all customers choose individually
optimal strategies. The equilibrium attained is a Wardrop equilibrium
that was first described in \cite{Wardrop52} and used extensively
in transportation systems. The Wardrop equilibrium routing condition
on $p^E(\beta)$ for all $\beta$ is that 
\begin{equation}
  \label{eq:wardrop}
  \begin{array}{lll}
     p^E(\beta) \geq 0 & \mbox{implies that} & c + \beta D_2 \leq X^E(\beta) + 
    \beta D_1(\beta).
    \end{array}
 \end{equation}
Further $0 < p^E(\beta) < 1$ implies $c + \beta D_2 = X^E(\beta) + 
    \beta D_1(\beta).$

%
%

The following theorem recalls the equilibrium strategy $(p^E(\beta), X^E(\beta))$
for the system model considered in \cite{Bodas14b}. The key difference between 
the models is that the FIFO server in \cite{Bodas14b} charges no admission price
and a customer joining the HBF server is required to pay a minimum bid $M.$  
 
\begin{theorem}
  \label{thm:threshold-equilibrium} 
  Using $\beta_1$ determined below, define $p^E(\beta),$ $F_1(\beta),$ and $W_0$
  as follows.
  \begin{eqnarray}
    p^E(\beta) & = & \begin{cases} 
      0 & \mbox{ for } \beta > \beta_1,\\
      t & \mbox{ for } \beta = \beta_1,\\
      1 & \mbox{ for } \beta < \beta_1.
    \end{cases} 
    \label{eq:P_E}
    \\
    F_1(\beta) & := & 
    \begin{cases} 
      0 & \beta \ < \beta_1\\ 
      \frac{\int_{\beta_1}^\beta dF(x) }{\int_{\beta_1}^b dF(x) }
      & \beta_1 \ \leq \beta \ \leq b, \\ 1 & \beta > b, 
    \end{cases} \label{eq:F_1}
    \\
    W_0 & = & \frac{\lambda_1}{2} \int_0^\infty \tau^2 dG(\mu_1 \tau), 
    \label{eq:W_0} \\
    X^E(\beta) & = & \int_0^\beta \frac{2 \rho_1 W_0 y}{\left( 1 - \rho_1 +
        \rho_1 F_1(y)\right)^3}\ dF_1(y)
    \label{eq:X_E}
  \end{eqnarray}
  For the routing and bidding policy $(p^{E}(\beta), X^{E}(\beta)),$
 $\beta_1$ is determined as follows.
  \begin{itemize}
  \item If using $\beta_1=a$ in \eqref{eq:P_E}--\eqref{eq:X_E} satisfies
    $M + a D_1(a) < a D_2,$ then set $\beta_1=a.$
  \item Else if using $\beta_1=b$ in \eqref{eq:P_E}--\eqref{eq:X_E}
    satisfies $M + b D_1(b) > b D_2$ then set $\beta_1=b.$
  \item Else find $\beta_1 \in (a,b)$ which when used in
    \eqref{eq:P_E}--\eqref{eq:X_E} satisfies 
    \begin{equation}
      M + \beta_1 D_1(\beta_1) = \beta_1 D_2.
      \label{eq:v_1-defn}
    \end{equation}
  \end{itemize}
 
  $(p^E(\beta), X^E(\beta))$ is an equilibrium strategy with $\beta_1$ defined
  as above. Further, $\beta_1$ is unique. \qed
\end{theorem}

It is important to point out that the  Wardrop 
condition for the case with a minimum bid $M > 0$ and a free 
FIFO server ($c=0$) that is used in the above theorem is that
\begin{equation*}
  \begin{array}{lll}
    p^E(\beta) \geq 0, & \mbox{implies that} & \beta D_2 \leq M + X(\beta) + 
    \beta D_1(\beta).
    \end{array}
 \end{equation*}

\begin{remark}
  Note from the conditions for $\beta_1$ that when 
$\beta_1 \in (a, b),$ the expected cost experienced by a customer with
$\beta = \beta_1$ at the two servers is same and hence the value of $t$ 
in Eq. \eqref{eq:P_E} is arbitrary. Also, $F_1(\cdot)$ in the theorem 
corresponds to the type profile of those customers that choose to join
the HBF server.
\end{remark}
The above theorem determines the equilibrium routing and bidding pair 
$(p^E(\beta), X^E(\beta))$ where $p^E(\beta)$ is characterized by a unique
threshold $\beta_1.$ It is natural therefore to ask whether a similar 
single threshold routing equilibrium holds when $M=0$ and the FIFO server charges an 
admission price. This question is analyzed in the next section. 

\section{Characterizing the Wardrop Equilibrium}
\label{sec:main}
In this section, we return to our main model and drop the assumption 
of a free FIFO service and minimum bid $M$ made in \cite{Bodas14b}. To 
simplify the analysis and restrict
the number of cases arising due to the heterogeneity of servers, we shall 
henceforth assume that $\mu_1 = \mu_2 = \mu,$ i.e., the servers are 
identical with service rate $\mu$. We shall further assume that $c > 0$ 
and that the minimum bid $M = 0.$ 

A possible candidate for the Wardrop equilibrium routing and bidding policy 
for this case  is the one that was identified in 
Theorem \ref{thm:threshold-equilibrium}. Recall that such a policy 
was characterized by a threshold $\beta_1 \in [a, b]$ and 
customers with $\beta > \beta_1$ choose the HBF server while those with 
$\beta < \beta_1$ choose the FIFO server at equilibrium.
We begin the analysis of this section by  first investigating whether the
routing and bidding policy $p^E(\beta)$ and $X^E(\beta)$ as described
in Eq. \eqref{eq:P_E}--\eqref{eq:X_E} of Theorem \ref{thm:threshold-equilibrium} 
holds for some $\beta_1 \in [a,b]$.  In the following lemma, we will show that 
when $c > 0$ and $M=0,$ $p^E(\beta)$ and $X^E(\beta)$ as described
in Eq. \eqref{eq:P_E}--\eqref{eq:X_E} with $\beta_1 \in (a, b]$ is not possible.
Further if $p^E(\beta)$ and $X^E(\beta)$ satisfy Eq.
\eqref{eq:P_E}--\eqref{eq:X_E}, then the underlying threshold $\beta_1$ 
must satisfy $\beta_1 = a,$ and the corresponding equilibrium strategy for 
all customers is to choose the HBF server.

\begin{lemma}
 \label{lem:two_thresh}
Consider a routing and bidding policy $(p^{E}(\beta), X^{E}(\beta))$ 
as described in Eq. \eqref{eq:P_E}--\eqref{eq:X_E} of Theorem \ref{thm:threshold-equilibrium}
and assume that $c > 0$ and $M = 0.$  Then such a $(p^{E}(\beta), X^{E}(\beta))$ with 
$\beta_1 \in (a, b]$ and satisfying the Wardrop equilibrium 
conditions of Eq. \eqref{eq:wardrop} does not exist.
\end{lemma}
 \begin{proof}
 We first prove that $\beta_1 \notin (a,b)$ using contradiction.
 Assume that when $c > 0$ and $M = 0,$ the corresponding 
 $(p^{E}(\beta), X^{E}(\beta))$  satisfy Eq. \eqref{eq:P_E}--\eqref{eq:X_E}
 with $\beta_1 \in (a,b).$
 Since $\beta_1 \in (a,b),$ 
 the equilibrium routing function $p^E(\beta)$  requires that 
 $c + \beta_1 D_2 =  \beta_1 D_1(\beta_1).$ This implies that 
 $D_2 < D_1(\beta_1)$ and hence for any $\beta < \beta_1,$ we have 
 
 \begin{eqnarray}
 \label{eq:lem_two_th}
 c & = &  \beta_1 (D_1(\beta_1) - D_2) \nonumber \\
  & > & \beta (D_1(\beta_1) - D_2) \nonumber \\
  & = & \beta (D_1(\beta) - D_2). 
 \end{eqnarray}
 To see how the last equality is true, consider a customer with
 $\beta <  \beta_1$ (that deviates from the prescribed $p^E(\beta)$) and 
 chooses the HBF server instead of the FIFO  server. Note from the 
 definition of $F_1(\cdot)$ and $X^E(\cdot)$ in Eq. \eqref{eq:F_1}
 Eq. \eqref{eq:X_E} respectively that
 $X^E(\beta_1)= 0.$ Using the property from \cite{Glazer86,Bodas14b}
 that $X^E(\beta)$ is increasing in $\beta$ and the fact that
 the $\beta$ customer is infinitesimal, we can conclude that
 $X^E(\beta)= 0.$ The deviating $\beta$ customer therefore does not
 pay a bid and will occupy the end of the queue. The infinitesimal 
 nature of this customer does not affect the delay at any of the two 
 servers and hence for such a deviating customer with $\beta < \beta_1$
 we have  $D_1(\beta) = D_1(\beta_1).$ 
 From Eq. \eqref{eq:lem_two_th}, we have $c + \beta D_2 > \beta D_1 (\beta).$
 Thus for this customer, the cost at the FIFO queue is more than the
 cost it would experience at the HBF server. This customer has an incentive
 to deviate from $p^E(\beta)$ and hence $p^E(\beta)$ is not an equilibrium  
 for any $\beta_1 \in (a,b).$

 Now consider the case when $\beta_1  = b.$ This implies that none of the 
 customers  choose the HBF server and we have $c + \beta D_2 <  \beta 
 D_1(\beta)$ for all $\beta.$ Note that since $\beta_1  = b,$ $\lambda_1 = 0$
 and hence if a customer of type $\beta$ were to choose HBF, the 
 delay would have been $D_1(\beta) = \frac{1}{\mu}.$ At the FIFO server, we have 
 $\lambda_2 = \lambda$ and hence $D_2 = \frac{1}{\mu - \lambda_1}.$ 
 Now $\beta_1 = b$ and hence $c + \beta D_2 <  \beta 
 D_1(\beta)$ implies that $c + \frac{\beta }{\mu - \lambda_1} <  
 \frac{\beta}{\mu}$ for all $\beta.$ However this is not possible for any $\beta$
 as $c > 0$ and $\frac{1}{\mu - \lambda_1} > \frac{1}{\mu}.$
 
 Next, we outline the conditions on the parameters for $\beta_1 = a.$
 From the definition of $p^E(\beta),$ this corresponds to the case 
 when none of the customers choose the FIFO server i.e., when 
 $c + \beta D_2 > X(\beta) + \beta D_1(\beta)$ for all $\beta.$ 
 It is straightforward to see that this case is possible when $c$ is set to any
 arbitrarily large value such that for all $\beta$ we have 
 $c + \beta D_2 > X(\beta) + \beta D_1(\beta).$
This completes the proof.
 \end{proof}

An alternative candidate for the equilibrium routing and bidding policy
is as follows. 

\begin{eqnarray}
\label{eq:candidate_routing}
     \hat{p}(\beta) & = & \begin{cases} 
      1 & \mbox{ for } \beta > \beta_1,\\
      t & \mbox{ for } \beta = \beta_1,\\
      0 & \mbox{ for } \beta < \beta_1.
  \end{cases} 
 \end{eqnarray}
 
 with the threshold $\beta_1 \in (a, b)$ and the 
 corresponding distribution $F_1(\beta)$ satisfying
 \begin{eqnarray}
 \label{eq:candidate_F}
   F_1(\beta) & := & 
    \begin{cases} 
      0 & \beta \ < a\\ 
      \frac{\int_{a}^\beta dF(x) }{\int_a^{\beta_1} dF(x) }
      & a \ \leq \beta \ \leq \beta_1, \\ 1 & \beta > \beta_1. 
    \end{cases} \label{eq:F_1}
 \end{eqnarray}
 
It is easy to see that $\hat{p}(\beta)$ for $\beta_1 = a~ ( \mbox{~resp.~}b)$
is equal to $p^E(\beta)$ (Eq. \eqref{eq:P_E}) with  
$\beta_1 = b~ (\mbox{~resp.~}a)$ and the conditions for
such equilibria with $\beta_1 = a~(\mbox{~or~} b)$ is already
outlined in the previous lemma.
In the following lemma, we shall now show that when $c > 0$
and $M = 0,$ a routing equilibrium satisfying Eq. \eqref{eq:candidate_routing} 
with $\beta_1 \in(a,b)$ is not possible.

 \begin{lemma}
 Let $c > 0, M = 0$ and consider a routing policy that 
 satisfies Eq. \eqref{eq:candidate_routing} with $\beta \in (a, b)$
 and has the type profile of customers to the HBF server given by Eq.
 \eqref{eq:candidate_F}. Such a routing policy does not 
 satisfy the Wardrop equilibrium conditions of Eq. \eqref{eq:wardrop}.
 \end{lemma}
 \begin{proof}
The proof is by contradiction. Suppose $\hat{p}(\beta)$ as in   
Eq. \eqref{eq:candidate_routing} satisfies the conditions of Eq. \eqref{eq:wardrop}.
Since $\beta \in (a, b),$ we have 
\begin{equation*}
 c + \beta_1 D_2 = X(\beta_1) + \beta_1 D_1(\beta_1).
\end{equation*}
Now since $X(\beta)$ is increasing in $\beta$ and all customers with
$\beta < \beta_1$ choose the HBF server, we have $X(\beta_1) > 0$
and $D_1(\beta_1) = \frac{1}{\mu}.$ Further, $\lambda_2 > 0$ implies 
 $D_2  > \frac{1}{\mu} = D_1(\beta_1)$ and hence $X(\beta_1) > c.$ For $\beta > \beta_1,$
we therefore have the following. 
\begin{eqnarray}
X(\beta_1) - c  &=& \beta_1 \left(D_2 - \frac{1}{\mu}\right) \nonumber \\ 
& < & \beta \left(D_2 - \frac{1}{\mu}\right). \nonumber
\end{eqnarray}
Thus for any customer with $\beta > \beta_1$ we have 
\begin{equation}
\label{eq:lem2_thwo_thld}
 X(\beta_1) + \frac{\beta}{\mu} < c + \beta D_2.
\end{equation}
Now consider a customer with $\beta > \beta_1$ that deviates from
$\hat{p}(\beta)$ and chooses the HBF server instead of the FIFO server.
The bid paid by this marginal customer is $X(\beta) = X(\beta_1)$ 
resulting in $D_1(\beta) = \frac{1}{\mu}.$ From Eq. \eqref{eq:lem2_thwo_thld},
the cost at the HBF server for such a customer $(X(\beta_1) + \frac{\beta}{\mu})$
is lower than the corresponding cost $(c + \beta D_2)$ at the FIFO server.
This customer now has an incentive to deviate and the therefore the routing
policy $\hat{p}(\beta)$ of Eq. \eqref{eq:candidate_routing} is not an equilibrium 
policy.
\end{proof}

We now come to the main result of this section. Having discarded
two candidate policies, we shall provide the routing and bidding
policy satisfying the Wardrop equilibrium condition of Eq. \eqref{eq:wardrop}.
As we shall show in the following theorem, this policy is characterized by two thresholds 
$\beta_1$ and $\beta_2$ that satisfy $a < \beta_1 < \beta_2 < b.$ In this 
policy, customers with $\beta$ satisfying $\beta_1 < \beta < \beta_2$ choose 
the FIFO server while the rest choose the HBF server. We will also characterize the 
conditions for such an equilibrium policy in the following theorem.

\begin{theorem}
  \label{thm:two-threshold-equilibrium} 
Let $\beta_1$ and $\beta_2$ be two thresholds satisfying the 
conditions outlined below. Define $p^E(\beta),$ $F_1(\beta),$ 
and $W_0$ as follows.
  \begin{eqnarray}
    p^E(\beta) & = & \begin{cases} 
      1 & \mbox{ for } \beta_1 < \beta < \beta_2,\\
      t & \mbox{ for } \beta \in \{\beta_1, \beta_2\},\\
      0 & \mbox{~elsewhere~}.
    \end{cases} 
    \label{eq:P_E_tt}
    \\
    F_1(\beta) & := & 
    \begin{cases} 
      \frac{\int_{a}^\beta dF(x)}{\left(1 - \int_{\beta_1}^{\beta_2} dF(x)\right)} & \beta \ < \beta_1\\ 
      \frac{\int_{a}^{\beta_1}dF(x)}{\left(1 - \int_{\beta_1}^{\beta_2} dF(x)\right)}
      & \beta_1 \ \leq \beta \ \leq \beta_2, \\
  \frac{\int_{a}^{\beta_1}dF(x) + \int_{\beta_2}^{\beta}dF(x)}{\left(1 - \int_{\beta_1}^{\beta_2} dF(x)\right)}
      & \beta_2 \ < \beta \ \leq b,\\ 
 1 & \beta > b, 
    \end{cases} \label{eq:F_1_tt}
    \\
    W_0 & = & \frac{\lambda_1}{2} \int_0^\infty \tau^2 dG(\mu \tau), 
    \label{eq:W_0} \\
    X^E(\beta) & = & \int_0^\beta \frac{2 \rho_1 W_0 y}{\left( 1 - \rho_1 +
        \rho_1 F_1(y)\right)^3}\ dF_1(y)
    \label{eq:X_E_tt}
  \end{eqnarray}
  
  The conditions that need to be satisfied by $\beta_1$ and 
 $\beta_2$ are as follows.
\begin{enumerate}
\item $a < \beta_1 < \beta_2 < b.$
 \item $X^E(\beta_1) = c.$
\item $X^E(\beta_1) + \beta_1 D_1(\beta_1) = c+ \beta_1 D_2(\lambda_2)$
 where $\lambda_2 = \int_{\beta_1}^{\beta_2}\lambda dF(x).$
\end{enumerate}

Then, $(p^E(\beta), X^E(\beta))$ is an equilibrium strategy satisfying the 
Wardrop conditions of Eq. \eqref{eq:wardrop}.
\end{theorem}
\begin{proof}
 As in the case of Theorem \ref{thm:threshold-equilibrium}, we will prove
 that $p^E(\beta)$ and $X^E(\beta)$ as described in \eqref{eq:P_E_tt}--\eqref{eq:X_E_tt}
 satisfy the Wardrop condition.  Note that with 
 $p^{E}(\beta)$ as in \eqref{eq:P_E_tt}, the arrival rate to the FIFO server is
 $\lambda_2 = \lambda \int_{\beta_1}^{\beta_2} dF(\tau).$ The arrival rate to the 
 HBF server is $\lambda - \lambda_2$ and the type profile of customers choosing
 this server will be as in Eq. \eqref{eq:F_1_tt}. As in Theorem 
 \ref{thm:threshold-equilibrium},  customers that join the HBF server will use the
 equilibrium bidding policy of Eq. \eqref{eq:X_E_tt}. We will now verify that 
 $p^E(\beta)$ of \eqref{eq:P_E_tt} satisfies the  corresponding Wardrop
 condition. 

 First note that second and third conditions of the theorem imply that 
 $D_1(\beta_1) = D_2(\lambda_2).$ Further, since  $F_1(\beta)$ is a constant
for $\beta \in (\beta_1, \beta_2),$ we have $X(\beta_1) = X(\beta_2)$
and $D_1(\beta_1) = D_1(\beta_2).$
Now consider a customer with $\beta < \beta_1$ that chooses the HBF server
for the routing policy of Eq. \eqref{eq:P_E_tt}. For this customer, we know
from (3) in Property~1 of \cite{Bodas14b} that $X(\beta)$ is the optimal bid
minimizing its individual cost i.e.,   
\begin{eqnarray*}
X(\beta) + \beta D_1(\beta) < X(\beta_1) + \beta D_1(\beta_1).
\end{eqnarray*}
 This implies that 
\begin{eqnarray*}
X(\beta) + \beta D_1(\beta) < c + \beta D_2(\lambda_2)
\end{eqnarray*}
and hence any customer with $\beta < \beta_1$ has no incentive to choose
the FIFO server. Similarly, for a customer with $\beta > \beta_2$ we 
have    
\begin{eqnarray*}
X(\beta) + \beta D_1(\beta) & < & X(\beta_2) + \beta D_1(\beta_2) \\
& = & c + \beta D_2(\lambda_2)
\end{eqnarray*}
and clearly customers with $\beta > \beta_2$ have no incentive to choose
the FIFO server.

Now consider any customer with $\beta \in (\beta_1, \beta_2).$  
If such a customer were to choose the HBF server instead of the 
prescribed FIFO server, its cost at the HBF server will be  
$X(\beta_1) + \beta D_1(\beta_1).$ This follows from (1) and (3) of 
Property~1 in \cite{Bodas14b} and the fact that a deviation 
by a marginal customer does not change Eq. \eqref{eq:P_E_tt} and 
\eqref{eq:F_1_tt} due to the absolute continuity of $F.$
Now $X(\beta_1) + \beta D_1(\beta_1) = c + \beta D_2(\lambda_2),$
i.e., the cost at the two server is the same and there is no 
incentive for this customer to deviate from $p^E(\beta)$ given 
by Eq. \eqref{eq:P_E_tt}. 
This completes the proof.
\end{proof}
 
 \begin{remark}
In the above discussion, we have assumed $c > 0$ and $M=0.$
W.l.o.g  the observations are true even for the case when $c > M$
and $M \neq 0.$ For this case, let $\hat{c} = c- M.$ The above 
theorem and the previous lemmas will hold with $c$ being replaced
by $\hat{c}.$ It is easy to see that the case $c < M$ and 
$M \neq 0$ corresponds to the Wardrop equilibrium as described in
Theorem \ref{thm:threshold-equilibrium}.
 \end{remark}

We now provide an example to verify the claim in the above theorem.\\
\textbf{Example:} Consider an HBF and FIFO server with identical 
service rates $\mu = 5.$ The arrival rate for the customers is 
$\lambda = 4.$ $F(\cdot)$ is a uniform distribution with support $[a, b]$
 where $a = 0$ and $b = 10.$ Now suppose that the admission price
at the FIFO server is $c = 0.2017.$ Assume that the service 
requirement of all customers is identical with an exponential 
distribution of unit mean. For a system with the specified parameters,
the interest is to characterize the Wardrop equilibrium routing function. 
It can be numerically verified that corresponding $p^{E}(\beta)$
is characterized by the two thresholds $\beta_1 = 1.67$ and $\beta_2 = 5.66.$
With these thresholds, the corresponding arrival rates to the two servers 
are $\lambda_1 = 1.596$ and $\lambda_2 = 3.404.$ Due to the exponential 
service requirement, it is easy to see that $D_2 = \frac{1}{\mu - \lambda_1} = 0.2938.$ 
From the definition of $W_0$ and Eq. \eqref{eq:Wbeta}, the expected waiting 
time of customers in HBF queue expressed as a function of $\beta$ is 
$$
W_1(\beta) = \frac{1}{(\mu - \lambda(1-F_1(\beta)))^2}. 
$$
It can be verified that the value of $D_2(\beta_1) = W_1(\beta_1) + 
\frac{1}{\mu} = 0.2938.$ Further, evaluating $X^E(\beta_1)$ using Eq. \eqref{eq:X_E_tt}
gives us $X^E(\beta_1) = c = 0.2017.$ Clearly, all the three conditions
are satisfied by the choice of $\beta_1$ and $\beta_2$ validating the 
theorem.


\section{Discussion}
 
The present work characterizes a two threshold equilibrium
when the service system has parallel FIFO and HBF servers. 
It should be noted that our analysis holds not  
only in the case of the FIFO scheduling policy but also for policies
such as processor sharing or last-in, first-out (LCFS) or 
any other policy that does not differentiate between its customers
on the basis of their $\beta.$ 
As in \cite{Bodas14b}, the  application for this model can be 
in the form of a new pricing and auction mechanism in on-demand 
resource provisioning, e.g., cloud computing systems.

Some interesting extensions present themselves. If the two 
servers belong to the same system operator, it would be interesting to
know the revenue maximizing admission price $c$ at the 
FIFO server. On the other hand, if the FIFO server is perceived as
a competing system, then the revenue optimal $c$ obtained 
in this case would certainly be different.

\bibliographystyle{IEEEbib}
\bibliography{bodas}
 
\end{document}